\documentclass[11pt]{article}
\usepackage[T1]{fontenc}
\usepackage{lmodern}
\usepackage[margin=1in]{geometry}
\usepackage{amsmath,amssymb,amsthm,mathtools}
\usepackage{enumitem}
\usepackage{microtype}
\usepackage[colorlinks=true,linkcolor=blue,citecolor=blue,urlcolor=blue]{hyperref}
\usepackage{aliascnt}

\theoremstyle{plain}
\newtheorem{theorem}{Theorem}[section]
\newaliascnt{lemma}{theorem}
\newtheorem{lemma}[lemma]{Lemma}
\aliascntresetthe{lemma}
\newaliascnt{corollary}{theorem}

\aliascntresetthe{corollary}
\theoremstyle{definition}
\newaliascnt{definition}{theorem}
\newtheorem{definition}[definition]{Definition}
\aliascntresetthe{definition}
\theoremstyle{remark}
\newaliascnt{remark}{theorem}

\aliascntresetthe{remark}
\theoremstyle{plain}
\newaliascnt{conjecture}{theorem}
\newtheorem{conjecture}[conjecture]{Conjecture}
\aliascntresetthe{conjecture}

\usepackage[nameinlink,capitalize]{cleveref}
\crefname{theorem}{Theorem}{Theorems}
\crefname{lemma}{Lemma}{Lemmas}
\crefname{corollary}{Corollary}{Corollaries}
\crefname{definition}{Definition}{Definitions}
\crefname{remark}{Remark}{Remarks}
\crefname{conjecture}{Conjecture}{Conjectures}
\Crefname{theorem}{Theorem}{Theorems}
\Crefname{lemma}{Lemma}{Lemmas}
\Crefname{corollary}{Corollary}{Corollaries}
\Crefname{definition}{Definition}{Definitions}
\Crefname{remark}{Remark}{Remarks}
\Crefname{conjecture}{Conjecture}{Conjectures}

\DeclareMathOperator*{\E}{\mathbb E}
\DeclareMathOperator{\tr}{tr}
\newcommand{\R}{\mathbb R}
\newcommand{\1}{\mathbf 1}
\newcommand{\calE}{\mathcal E}
\newcommand{\calF}{\mathcal F}

\newcommand{\eps}{\varepsilon}
\newcommand{\norm}[1]{\left\lVert #1\right\rVert}
\numberwithin{equation}{section}
\newcommand{\calX}{\mathcal{X}}
\DeclareMathOperator{\codim}{codim}
\def\cE{{\cal E}}
\DeclareMathOperator{\dist}{dist}
\DeclareMathOperator{\Var}{Var}
\def\bone{\bf 1}
\title{High-Dimensional Expanders, the Sparsest Cut Problem, \\and Steurer's Conjecture}
\author{
\begin{tabular}[t]{c}
\Large Farzam Ebrahimnejad\\[3pt]
\small The Voleon Group\\[2pt]
\small {\texttt{f.ebrahimn@gmail.com}}
\end{tabular}
\and
\begin{tabular}[t]{c}
\Large Shayan Oveis Gharan\\[3pt]
\small University of Washington\\[2pt]
\small {\texttt{shayan@cs.washington.edu}}
\end{tabular}
}
\date{}

\begin{document}
\maketitle

\begin{abstract}
In 2010, Steurer conjectured that any family of $n$ unit-norm vectors $v_1,\dots,v_n$ with polynomially small average correlation $\E_{i,j}|\langle v_i,v_j\rangle|\leq n^{-\eps}$ contains linear-sized constant-separated sets. We refute this conjecture in a strong sense using the machinery of sparse high-dimensional expanders:  such vector families do not even have linear-sized $\frac{1}{\log^{1/4-o(1)}(n)}$-separated sets. Consequently, we show that there are families of vertex expanders on $n$ vertices for which the (average) $L_2$-mixing time to the uniform distribution of any reweighted simple random walk is at least $\log^{5/4-o(1)} n$.

\end{abstract}

\section{Introduction}
The landmark result of Arora-Rao-Vazirani \cite{ARV} which gave an $O(\sqrt{\log n})$-approximation for the sparsest cut problem motivated the use of the sum-of-squares hierarchy in approximation algorithms. At the heart of their proof they showed the following structural theorem:
\begin{theorem}[Arora-Rao-Vazirani \cite{ARV}]\label{thm:arv}
Given a set of unit-norm vectors $v_1,\dots,v_n$ that satisfy the squared triangle inequality, i.e., for any $1\leq i,j,k\leq n$, 
\begin{equation}\label{eq:triangleinequality}\norm{v_i-v_j}^2 \leq \norm{v_i-v_k}^2+\norm{v_k-v_j}^2.
\end{equation}
If, in addition, $\E_{i,j}\langle v_i,v_j\rangle^2 < o(1)$ (where the expectation is over the uniform distribution), then there are linear-sized sets $A,B\geq\Omega(n)$, such that for any $v_i\in A,v_j\in B$, $\norm{v_i-v_j}^2\geq \Omega(1/\sqrt{\log n}).$
\end{theorem}

Since then, a natural question puzzled researchers for decades: Can we get improved approximation algorithms for the sparsest cut problem, or unique games problems (or constraint satisfaction problems) shall we use more rounds of the  sum-of-squares relaxation?
Along this route, Barak-Raghavendra-Steurer \cite{BRS} proposed the global correlation rounding technique. Given a solution to the $k$-th round of the sum-of-squares hierarchy, if the ``average correlation'' between the given vectors is $\geq 1/k$, one can round the solution by conditioning on say $k$ random vertices. So, in order to get a sub-exponential time algorithm it is enough to find a rounding algorithm when the average correlation is say $<1/\sqrt{n}$. Motivated by this, and with the hope of extending the above theorem Steurer proposed the following conjecture in his thesis.

\begin{conjecture}[Steurer \cite{SteurerThesis}]\label{conj:steurer}
For every $\eps>0$ there are constants $\delta=\delta(\eps)>0$ and $\eta=\eta(\eps)>0$ such that the following holds for all sufficiently large $n$.  If $v_1,\ldots,v_n$ are unit vectors and
\[
        \E_{i,j} |\langle v_i,v_j\rangle|\le n^{-\eps},
\]
then there are linear-sized constant-separated sets, i.e., there are disjoint sets of vectors $A,B$ with $|A|,|B|\ge\delta n$ such that
\[
        \norm{v_i-v_j}_2^2\ge \eta
        \qquad\quad\forall v_i\in A,\ v_j\in B .
\]
\end{conjecture}

Steurer showed that the above conjecture, if true, implies a constant-factor $n^{n^\eps}$-time algorithm for the sparsest cut problem.
In addition, even if the conclusion holds with a better than $1/\sqrt{\log n}$-separation still it would lead to better approximation (in sub-exponential time).  

The above conjecture received significant attention in the community\footnote{See, for example, \url{https://tcsmath.wordpress.com/2013/06/26/separated-sets-in-unions-of-frames/}} as it seems the most natural approach to round solutions of the SOS hierarchy. 

A natural route to approach the above conjecture is to connect to the machinery of random walks and mixing time. 
\begin{definition}[Vertex Expansion]
    Given a graph $G=(V,E)$, the vertex expansion of $G$ is defined as 
    $$ h(G)=\min_{\emptyset \subsetneq S\subsetneq V, |S|\leq |V|/2} \frac{|N_G(S)|}{|S|},$$
    where $N_G(S)$ is the set of vertices $x\notin S$ that have at least one neighbor in $S$.
\end{definition}
\begin{conjecture}\label{conj:vertexexpansion}
    There are absolute constants $0<h<1, \beta>0$ such that for any graph $G=(V,E)$ with $n=|V|$ vertices and $h(G)>h$, there exists a symmetric stochastic matrix $P\in \R^{V\times V}$ supported on the edges of $G$ such that
    $$ \tr(P^{\beta \log n}) \leq O(1).$$
\end{conjecture}
It is not hard to see that the above two conjectures are equivalent. In the last section we explain a reduction from the first to the second.

\subsection{Main results}

In our main result we refute \cref{conj:steurer} in a strong sense.
\begin{theorem}\label{thm:main}
For every $0<\beta<1/2$, there are infinitely many $n>0$ and nonnegative unit vectors $v_1,\ldots,v_n$ satisfying squared triangle inequality \eqref{eq:triangleinequality} such that
\[
        \E_{i,j} |\langle v_i,v_j\rangle|\le n^{-1+o(1)},
\]
and every two disjoint sets of vectors $A,B$ with $|A|,|B|\ge\beta n$ contain $v_i\in A$ and $v_j\in B$ with
\[
        \norm{v_i-v_j}_2^2
        \le
        O_{\beta}\bigl((\log n)^{-1/4+o(1)}\bigr).
\]
\end{theorem}

We believe that the $\log^{-1/4}n$ separation above is not tight; perhaps, using a more optimized construction of sparse high dimensional expanders this can be improved to the tight bound of $\log^{-1/2}n$.
As an immediate consequence, our counterexample also refutes \cref{conj:vertexexpansion}.
\begin{theorem}\label{thm:vertex-mixing}
For every $0<h_*<1$, there is an infinite family of graphs
$H_m=(V_m,E_m)$, with $n_m=|V_m|\to\infty$, such that:
\begin{enumerate}[label=(\roman*),leftmargin=2.5em]
    \item $h(H_m)\ge h_*$ for every $m$;
    \item For every symmetric stochastic matrix $P$ supported on edges
    of $H_m$ with $n:=n_m$ vertices and $t \ge 1$, one has
    \[
     \tr(P^{4t})\geq n^{1-o(1)}e^{-t\log^{-1/4+o(1)}n}.
    \]
\end{enumerate}
\end{theorem}
The above theorem implies that the $L_2$-mixing time of the random walk $P$ which converges to the uniform stationary distribution on $H_m$ is at least $\log^{5/4-o(1)} n$; in other words, although a simple random walk always mixes in $O(\log n)$ on regular {\bf edge} expander, there are families of vertex expanders in which no reweighted (simple) random walk mixes faster than $\log^{5/4-o(1)}n$ steps to the uniform stationary distribution. 

This theorem should be contrasted with recent results on the reweighting eigenvalue problem \cite{OTZ,KLT}. The authors show that for any vertex expander $G$, there is a symmetric stochastic matrix $P$ such that $\lambda_2(P)\leq 1-\frac{h^2(G)}{\log \Delta}$, where $\Delta$ is the maximum degree of G. Such a result implies an upper-bound of $O(\log n\log\Delta)$ on the mixing time of $P$ in the worst case. To put it in context, our result shows that these bounds cannot be improved to $O(\log n)$ and some dependencies on the max degree is necessary.

\subsection{Acknowledgements}
The authors used GPT-5.5 Pro during the development of this work, primarily through an agentic system for long-horizon mathematical research developed by the first author and built on top of Codex and GPT models. The system was supplied with the authors' evolving research notes and was used, under the authors' guidance, to explore proof strategies. In particular, the connection between \cref{conj:steurer} and the high-dimensional expander literature emerged during these interactions.

The second author would like to thank James R Lee and David Steurer for insightful discussions during the long term course of this project. The second author's research is funded by an NSF grant CCF-2203541, a Simons Investigator Award 928589, and
a Lazowska Endowed Professorship in Computer Science \& Engineering.

\section{Preliminaries}
\subsection{Random Walks}
Given a (weighted) graph $G=(V,E,w)$ with (weighted) adjacency matrix $A$; let $D$ be the diagonal matrix of vertex (weighted) degrees and let $P=D^{-1}A$ be the transition probability matrix of the simple random walk on $G$ and $\mu$ the stationary distribution of the walk, i.e., we have $\mu^T P =\mu$ and $P\bone=\bone$.
 For two functions $f,g:V\to\R$, let
  $$\langle f,g\rangle_\mu= \E_{x\sim\mu} f(x)g(x).$$ 
 It turns out that $P$ is self-adjoint  w.r.t. this inner product. So, by spectral theorem it has real eigenvalues with an orthonormal set of eigenfunctions. The largest eigenvalue is 1 corresponding to the all-ones eigenfunction.
 
 Given $f,g:V\to\R$ the Dirichlet form is defined as follows:
 $$ \cE_P(f,g) = \frac12\E_{x\sim\mu} \sum_y P(x,y) (f(x)-f(y))(g(x)-g(y))$$
 The variance of $f$ is defined as follows:
 $$ \Var_\mu(f)=\frac{1}{2} \E_{x,y\sim\mu} (f(x)-f(y))^2$$
 \begin{lemma}[Rayleigh Quotient]\label{lem:rayleigh}
     The second eigenvalue of $P$ is
     $$ \lambda_2(P)=1-\min_{f:\langle f,\bone\rangle_\mu=0} \frac{{\cal E}_P(f,f)}{\langle f,f\rangle_\mu} = 1-\min_{f\neq \bone} \frac{\cE_P(f,f)}{\Var_\mu(f)}$$ 
 \end{lemma}

\subsection{Local Spectral Expanders}
A simplicial complex $\calX$ on a finite ground set $V$
is a downward-closed collection of subsets of $V$, i.e., if $\tau\in \calX$ and $\sigma \subseteq \tau$, then $\sigma \in \calX$. 
The elements of $\calX$ are called {\bf faces}, and the maximal faces are called {\bf facets}. 
We say that a face $\tau$ is of dimension $k$ if $|\tau| = k$ and write $\text{dim}(\tau) = k$\footnote{Note that this differs from the typical topological definition of dimension for faces of a simplicial complex.}. 
A simplicial complex $\calX$ is a {\bf pure} $d$-dimensional complex if every facet has dimension $d$. 

Given a  $d$-dimensional complex $\calX$,
for any   $0 \leq i \leq d$,    define  $\calX(i) = \{\tau \in \calX \mid \text{dim}(\tau) = i\}$, e.g., $\calX(1)=V$ is called the set of vertices of $\calX$.  Moreover, the codimension of a face $\tau \in \calX$ is defined as $\codim (\tau)  = d - \dim (\tau)$. For a face $\tau \in \calX$, define the  {\bf link} of $\tau$ as the simplicial complex $\calX_{\tau} = \{\sigma \setminus \tau \mid \sigma \in \calX, \sigma \supseteq \tau\}$.

Given $\calX$, the degree of $\calX$ is defined as 
$$ \max_{x\in V} |\{\sigma\in X(d): x\in \sigma\}|.$$
 We say $\calX$ is a {\bf connected} complex if for every $\tau$ of codimension at least 2, the 1-skeleton of $\calX_\tau$ is a connected graph.

\begin{definition}[(Top-link) Spectral Expanders]
Consider a link $\tau$ of co-dimension 2. The 1-skeleton of $X_\tau$ is a simple graph with vertex set $\calX_\tau(1)$ and edge set $\calX_\tau(2)$.
    We define $\lambda_2(\tau)$ as the 2nd-eigenvalue of the simple random walk on the 1-skeleton of $\calX_\tau$.
    We say $\calX$ is an $\alpha$-top-link expander if $\lambda_2(\tau)\leq\alpha$ for any $\tau$ of co-dimension 2.
\end{definition}

\paragraph{The Down-up walk.} Given a $d$-dimensional complex $\calX$ and a facet $\tau\in \calX(d)$, one can run a Markov chain called  the down-up walk on $\calX(d)$: suppose we are currently at $\tau$; we choose a vertex $v \in \tau$  uniformly at random, and among all facets $\sigma\supseteq \tau \setminus \{v\}$ we move to one uniformly at random. 
 Note that this walk is lazy as we can go back to $\sigma$; furthermore, if we jump to $\tau$, we always have 
\begin{equation*}\label{eq:symmdiff}|\tau\Delta\sigma|:=|\tau\smallsetminus \sigma|+|\sigma\smallsetminus\tau| \in \{0, 2\}.
\end{equation*}
It turns out that if $\calX$ is connected, this Markov chain converges to a uniform distribution over all facets of $X$, call this distribution $\mu$. We write $P^\vee_d$ to denote the transition probability matrix of the down-up walk (with stationary distribution $\mu$).

Over the last few years (top-link) high-dimensional expanders have been used extensively in the analysis of Markov chains.  The following local-to-global theorem is central in many such applications: 
\begin{theorem}[\cite{DinurKaufman, Oppenheim18,AlevLau}]\label{thm:oppenheims}
If $\calX$ is a connected $d$-dimensional $\frac{1-\delta}{d}$-top-link spectral expander then the spectral gap of the down-up walk matrix ($P^\vee_d$) is at least $\frac{(1-\frac{1-\delta}{\delta d})^d}{d}$. In particular, for $\delta=1/2$, we get a lower bound of $\frac{1}{e\cdot d}$.
\end{theorem}

\subsection{Sparse High-Dimensional Expanders}
Let $\calX$ be a $d$-dimensional $\alpha$-top link expander for some $\alpha<1/d$ with degree $\Delta$. We say $\calX$ is {\bf sparse} if $\Delta$ is a sub-polynomial of $\calX(1)$, ideally a constant. 
Over the past years numerous constructions of sparse top-link expanders have been given \cite{LSVExplicit, LSV, EvraKaufman, DinurKaufman, KaufmanOppenheim18, ODonnellPratt22, HR26}. They have found numerous applications, e.g., in the design and analysis of locally testable codes \cite{DELLM22} or in construction of near-optimal PCPs \cite{BMVY25}. 
It turns out that most of these constructions work for our purpose of proving \cref{thm:main}. Here, for simplicity of the argument we explain a very recent combinatorial construction by Hopkins and Ray \cite{HR26}. 

A $d$-dimensional complex $X$ is $d$-partite if its vertices $\calX$ can be partitioned into $d$ sets, $P_1,\dots,P_d$ such that every facet has exactly one vertex from each part.

\paragraph{Projected Flags Complex.}
Let $r \in \mathbb{N}$, let $q$ be a prime power, and let
\[
S=\{s_1<s_2<\cdots<s_k\}\subseteq [2r-1].
\]
The projected flags complex $\calX^S$ is the $k$-partite $k$-dimensional
simplicial complex whose vertices are the subspaces
$W\subset \mathbb{F}_q^{2r}$ with $\dim(W)\in S$, partitioned according
to dimension. Its facets are maximal chains
\[
\calX^S(k)
=
\{\, W_1 \subset W_2 \subset \cdots \subset W_k
:\dim(W_i)=s_i \,\}.
\]

We take $S$ to be 
\[  
S_k=\{r-k,\ldots,r-1,r,r+1,\ldots,r+k\}.
\]
 Observe that there are  $q^{r^2\pm\Theta(kr)}$ subspaces of dimension $i$ for any $i\in S$, eliminating the imbalance up to lower order factors. Fixing
$k$ and letting $r\to \infty$, we get a strongly explicit infinite family of complexes with $q^{\Theta(r^2)}$ vertices
and degree $q^{O(kr)} \approx 2^{O(k\sqrt{\log \calX(1)})}$ which is sub-polynomial in the number of vertices.
Furthermore, it is not hard to see $\lambda_2(\tau)\leq \frac{2}{\sqrt{q}-1}$ for any $\tau$ of co-dimension 2. Summarizing and by a standard application of \cref{thm:oppenheims} with $\delta=1/2$ we obtain the following theorem:
\begin{theorem}\label{thm:sparse-hdx}
Given integers $1\le k<r$, let $d=2k+1$.  For every prime power
$q\ge (4d+1)^2$, there is a pure $d$-partite $d$-dimensional simplicial complex
$\calX$ with $q^{\Theta(r^2)}$ vertices, degree at most $q^{O(rd)}$ which is a $1/2d$-top-link-expander and  the  $\lambda_2(P^\vee_d)\leq 1-1/ed$.
\end{theorem}

\section{Local Spectral Expanders and the Main Technical Theorem}

In this section we prove our main technical theorem.

\begin{theorem}\label{thm:abstract-top-face}
Let $\calX$ be a pure $d$-dimensional complex with degree $\Delta$, and suppose the spectral gap of the $d$-dimensional down-up walk is at least $\frac{\alpha}{d}$.  Then there are 
$n:=|\calX(d)|$ nonnegative unit vectors
$v_1,\ldots,v_n\in \mathbb R_{\ge0}^{\calX(1)}$ satisfying squared triangle inequality \eqref{eq:triangleinequality} such that:
\begin{enumerate}[label=(\roman*),leftmargin=2.5em]
    \item \label{item-i}
    $\mathbb E_{i,j}|\langle v_i,v_j\rangle|\le \frac{\Delta}{n}$;
    \item \label{item-ii} for every $0<\beta<1$ and every two disjoint sets of vectors 
    $A,B$ with $|A|,|B|\ge \beta n$, there exist
    $v_i\in A$ and $v_j\in B$ such that
    \[
        \|v_i-v_j\|_2^2
        \le
        \frac{2}{\beta\sqrt{ \alpha\cdot d}}.
    \]
\end{enumerate}
\end{theorem}
\begin{proof}
List the top faces as $\calX(d)=\{\sigma_1,\ldots,\sigma_n\}$ and set
$v_i=d^{-1/2}\mathbf 1_{\sigma_i}$.  Observe that these vectors satisfy squared triangle inequality (since they are just of $0/1$-vectors up to rescaling). Then $v_i\in\mathbb R_{\ge0}^{\calX(1)}$
is a unit vector and
\[
        \langle v_i,v_j\rangle=\frac{|\sigma_i\cap\sigma_j|}{d}.
\]

First, we prove \ref{item-i}. For $x\in\calX(1)$, let $\calF_x=\{i:x\in\sigma_i\}$.  Since
$|\calF_x|\le\Delta$,
\[
\frac{1}{n^2} \sum_{i,j=1}^n |\langle v_i,v_j\rangle|
=
\frac1{dn^2}\sum_x |\calF_x|^2
\le
\frac{\Delta}{dn^2}\sum_x|\calF_x|
=
\frac{\Delta}{n},
\]
where the last identity uses that $\sum_x |\calF_x| = n\cdot d$, since every $\sigma_i$ has exactly $d$ many vertices (of $\calX(1)$). 

Now, we prove \ref{item-ii}. 
First, we construct a graph $G$ on vertex set $\calX(d)$ which is the support graph of the down-up walk: Namely, we connect $\sigma,\tau\in \calX(d)$  if $|\tau\Delta\sigma|=2$, i.e., if the walk can jump from $\sigma$ to $\tau$. For any $\sigma,\tau\in \calX(d)$, we let $\dist(\sigma\tau)$ to denote the shortest path distance of $\sigma$ and $\tau$ in $G$.

Fix $0<\beta<1$ and disjoint sets $A,B\subseteq[n]$ with
$|A|,|B|\ge\beta n$.  Suppose, for
contradiction, that $|\sigma\cap\tau|<d(1-\eps)$ for every $\sigma\in A$ and
$\tau\in B$ for  some $0<\eps<1$ that we fix later.  Therefore, for any $\sigma\in A,\tau\in B$, $|\sigma\Delta\tau|> 2\eps \cdot d.$ Consequently, $\dist(\tau,\sigma) > \eps\cdot d$.

Define $f:\calX(d)\to \R_{\geq 0}$ as follows: 
\[
        f(\sigma)=\max\left\{1+\delta-\frac{\dist
        (\sigma,A)}{\eps\cdot d},0\right\},\quad\quad  \forall \sigma\in\calX(d),
\]
where $\dist(\sigma,A)$ is the shortest path distance from $\sigma$ to any facet in $A$, and $\delta>0$ but arbitrarily close to 0.
We show that
\begin{equation}\label{eq:pointcareupper}
    \frac{\calE(f,f)}{\Var(f)} <\frac{1}{\beta^2\cdot (\eps\cdot d)^2}
\end{equation}
First, observe that by the above discussion, $f>1$ on all facets of $A$ and $f=0$ on all facets of $B$. Furthermore, 
$f$ is $\frac{1}{\eps \cdot d}$-Lipschitz with respect to the distance function.  
Since every step of the down-up walk is supported on $G$ and $|f(\sigma)-f(\tau)|\leq 1/(\eps d)$ for any edge $(\sigma,\tau)$. Therefore, 
$$\calE(f,f)\le \frac{1}{(\eps\cdot d)^2}.$$
On the other hand, since $f>1$ on $A$ and $f=0$ on $B$.
\[
        \Var(f)
        =
        \frac12\mathbb E_{\sigma,\tau\sim \calX(d)}
        \bigl(f(\sigma)-f(\tau)\bigr)^2
        \ge \frac{|A|}{|\calX(d)|}\cdot \frac{|B|}{{\calX(d)|}}
        > \beta^2,
\]
  This proves \eqref{eq:pointcareupper}. Therefore, since $f$ is not a constant function, by definition of the Poincar\'e constant, we have 
\[
    \frac{\alpha}{d}\leq \alpha(P^\vee) < \frac{1}{\beta^2(\eps\cdot d)^2}
\]
Now, letting $\eps = \frac{1}{\beta \sqrt{\alpha\cdot d}}$ we get a contradiction.

Thus, there exists $\sigma_i\in A$ and $\sigma_j\in B$ with $|\sigma_i\cap\sigma_j|\ge d(1-\eps)$.
For this pair,
\[
        \|v_i-v_j\|_2^2
        =
        2\left(1-\frac{|\sigma_i\cap\sigma_j|}{d}\right)
        \le
        2 \eps 
        =
        \frac{2}{\beta\sqrt{\alpha\cdot d}}.
\]
\end{proof}

\section{Proof of the Main Result}\label{sec:main-proof}

In this section we prove \cref{thm:main}.

\begin{proof}[Proof of \cref{thm:main}]
For a sufficiently large $k$, let $d=2k+1$ and choose $r$ such that $r^{1-\eps} = d$ for $\eps>0$ a very small constant.
Choose a prime (power)  $ (4r+1)^2\leq q\leq O(r^2)$, and apply
\cref{thm:sparse-hdx} to get a $d$-dimensional $d$ partite complex $\calX$.  Let $n=|\calX(d)| \geq q^{\Omega(r^2)}$. 
Then, the degree of $\calX$ is at most $\Delta\le q^{O(rd)}$.
Since $\lambda_2(P^\vee_d)\leq 1-\frac{1}{e\cdot d}$, by
\cref{thm:abstract-top-face} we obtain nonnegative unit vectors
$v_1,\ldots,v_n$ satisfying squared triangle inequality such that (i)
\[
        \E_{i,j}|\langle v_i,v_j\rangle|
        \le \frac{\Delta}{n}
        \leq 
        n^{-1+o(1)},
\]
where the last inequality uses $n\geq q^{\Omega(r^2)}$, $\Delta\leq q^{O(rd)}=q^{O(r^{2-\eps})}$ for some $\eps>0$. 
Furthermore by (ii) of \cref{thm:abstract-top-face}
for every two disjoint sets of vectors
$A,B$ with $|A|,|B|\ge \beta n$ there are vectors
$v_i\in A$ and $v_j\in B$ with
\[
        \|v_i-v_j\|_2^2\le \frac{2}{\beta\sqrt{ d/e}} = O(\log n)^{-1/4+O(\eps)}.
\]
where the last inequality uses that $n\leq |\calX(1)|\Delta/d=q^{O(r^2)}= q^{O(d^{2(1+2\eps)})} = 2^{O(d^{2(1+2\eps)}\log d)}$.
The statement follows by letting $\eps\to 0$. 
\end{proof}

We remark that, to the best of our knowledge, the construction of \cref{thm:sparse-hdx} is not necessarily tight; If one can get similar guarantees with a top-link expander which has say $q^r$ many vertices and degree at most $q^{\sqrt{r d}}$, then our bounds shall improve to $O(1/\sqrt{\log n})$ implying tightness of \cref{thm:arv}.

\section[Vertex expanders and reweighted L2 mixing]{Vertex Expanders and Reweighted $L_2$ Mixing}\label{sec:vertex-appendix}

In this section we prove \cref{thm:vertex-mixing}.  The reduction was already known by experts; we just explain it here for completeness.
The construction starts from the close-pair graph of the vectors in \cref{thm:main}.  
First, we use a simple peeling argument to turn this into an (induced) vertex expander.  Then, we use a trace obstruction: if a reweighting supported on this graph mixed too quickly, it would contradict the small average Gram mass of the vectors.
\begin{lemma}[Peeling]\label{lem:peeling}
For every $0<h<1$ there is $\beta>0$ such that the following holds: Given  an
$n$-vertex graph $G=(V,E)$ such that for any two disjoint sets $A,B$ with
$|A|,|B|\ge \beta n$ there is at least one edge in the induced cut $(A,B)$, $G$ contains an induced
subgraph $H$ on at least $n/2$ vertices with vertex expansion at least $h$.
\end{lemma}

\begin{proof}
 Let $\beta<1/4$ be a sufficiently small parameter that we choose later.
Starting with $W_0=V$, repeatedly remove a nonempty set
$U_i\subseteq W_i$ with $|U_i|\le \beta n$ and
$|N(U_i)\cap W_i|<h\cdot |U_i|$, together with $N(U_i)\cap W_i$. We stop  if no such  set exists or if $|W_i|<n/2$. We let $W$ be the final set.

First, we show that $|W|\ge n/2$. For contradiction suppose $|W|<n/2$. Since we stop the first time $|W_i|<n/2$, and since in each step we delete at most $(1+h)n\beta$ vertices, we must have $|W|\geq n(1/2-2\beta)\geq n/4$. On the other hand, by construction the deleted sets $U_i$'s have no edges to $W$ and the sum of their sizes is at least $n/4$ (since $h<1$). 
Letting $A$ be the union of deleted $U_i$'s and $B=W$
gives a contradiction.

Now, assume $|W|\geq n/2$; We claim that the induced graph $G[W]$ satisfies the lemma's conclusion. First, since $|W|\geq n/2$, the algorithm stops because every set of size at most $\beta n$ has vertex expansion at least $h$ in $G[W]$. To finish the proof, consider a set $U$ with $\beta n<|U|\le |W|/2$.  If
$|N(U)\cap W|<h\cdot |U|$, then $W\smallsetminus U\smallsetminus N(U)$ has size
at least $(1-h)\frac{|W|}{2}\ge \frac{1-h}{4}n$.
Letting $\beta\leq(1-h)/4$, the sets $A=U,B=W\smallsetminus U\smallsetminus N(U)$ have size at least $\beta n$
but no edges connects them which is a contradiction.
\end{proof}

\begin{lemma}[Trace Obstruction]\label{lem:trace}
Let $v_1,\ldots,v_n$ be unit vectors with Gram matrix $A$, i.e., $A_{i,j}=\langle v_i,v_j\rangle$.
Let $0\le \delta<2$, and let $G_\delta$ be the graph joining pairs with $\|v_i-v_j\|_2^2\le \delta$.  If $P$ is a symmetric stochastic matrix with nonnegative entries supported on $G_\delta$ together with self-loops, then
\[
        \left(1-\frac{\delta}{2}\right)^{2t}\le \frac1n \left(\norm{A}_F^2\cdot  \tr(P^{4t})\right)^{1/2} \,.
\]
\end{lemma}

\begin{proof}
First notice if $\norm{v_i-v_j}^2\leq \delta$, then $\langle v_i,v_j\rangle\geq \1-\delta/2$, since $v_i,v_j$ are unit-norm. Therefore, since $P$ is supported on $G_\delta$,
\begin{equation}\label{eq:trAP}
\frac1 n\tr(AP)\ge 1-\delta/2.
\end{equation}
Since $P$ is symmetric, by the spectral theorem we can write, $P=\sum_{i=1}^n\lambda_iu_iu_i^{\mathsf T}$ where $u_1,\dots,u_n$ form an orthonormal family of vectors.  We write,
\[
\left(\frac1n\tr(AP)\right)^{2t}
= \left(\frac1n \sum_{i=1}^n \lambda_i u_i^TA u_i\right)^{2t} \leq \left(\frac1n \sum_{i=1}^n |\lambda_i| u_i^TA u_i\right)^{2t}
\]
Since $A$ is PSD, $u_i^{\mathsf T} A u_i \geq 0$ for all $i$. Furthermore, 
$$\sum_{i=1}^N u_i^{\mathsf T} A u_i \sum_{i=1}^n \tr(u_i u_i^{\mathsf T} A)=\tr(A)=n.$$ 
The last identity uses that $A_{i,i}=\norm{v_i}^2=1$ for all $i$.
 Thus, it follows that $\{\frac1n u_i^{\mathsf T}A u_i\}_{1\leq i\leq n}$ form a probability distribution. Having this we can apply Jensen's inequality to get
$$ \left(\frac1n\tr(AP)\right)^{2t} \leq \left(\frac1n \sum_{i=1}^n |\lambda_i|^{2t} u_i^{\mathsf T}A u_i\right) = \frac1n \tr(AP^{2t}) \underset{\text{Cauchy-Schwarz}}{\le}\frac1n \sqrt{\tr(A^2)\tr(P^{4t})},
$$
where in the equality we use that $|\lambda_i|^{2t}=\lambda_i^{2t}$
Lastly, the statement follows since $\tr(A^2)=\norm{A}_F^2$ putting together with \eqref{eq:trAP}.
\end{proof}

\begin{proof}[Proof of \cref{thm:vertex-mixing}]
Given $h_*$, let $\beta>0$  be the constant given by \cref{lem:peeling}
with $h=h_*$. By \cref{thm:main} with density parameter $\beta$ we obtain  unit vectors
$v_1,\ldots,v_n$ such that (i)  $\sum_{i,j}|\langle v_i,v_j\rangle|\leq n^{1+o(1)},$
and (ii) for every
two disjoint subsets $A,B$ of vectors of size at least $\beta n$ there is a pair 
of vectors $v_i\in A,v_j\in B$ with
\[
    \norm{v_i-v_j}_2^2\leq (\log n)^{-1/4+o(1)}=:\delta.
\]

Let $G$ be the graph on $[n]$ joining $i,j$ if and only if
$\|v_i-v_j\|_2^2\le\delta$.  Then $G$ has no two disjoint sets of size at
least $\beta n$ with no edges between them.  So, \cref{lem:peeling}, $G$
contains an induced subgraph $H=G[W]$ with $m:=|W|\ge n/2$ vertices and
$h(H)\ge h_*$.

Restrict the vectors to $W$.   we have 
$$ \sum_{i,j\in W}|\langle v_i,v_j\rangle|\leq n^{1+o(1)}.$$
Also, every edge of $H$ joins two vectors at squared
distance at most $\delta$.

Let $P\in \R^{W\times W}$  be an arbitrary symmetric stochastic matrix  supported on $G[W]$ together with self-loops.  By \cref{lem:trace} with
$\delta$ gives
\[
        \left(1-\frac{\delta}{2}\right)^{4t}
        \le \frac1{m^2}\tr(P^{4t})\cdot \sum_{i,j\in W}\langle v_i,v_j\rangle^2 \leq \frac1{m^2} \tr(P^{4t})\cdot \sum_{i,j}|\langle v_i,v_j\rangle|\leq \frac1{m^2}\tr(P^{4t})\cdot n^{1+o(1)},
\]
where the second inequality uses that $\norm{v_i}=1$ for all $i$.
Simplifying and  using $\delta=\log^{-1/4+o(1)}n$ we obtain that 
$$ \tr(P^{4t})\underset{1-x\geq e^{-2x}\text{ for }0<x<1/2}\geq m^2 n^{-(1+o(1))}e^{-t\log^{-1/4+o(1)}n} \geq m^{1-o(1)}e^{-t\log^{-1/4+o(1)}m}$$
where the last inequality uses $m\geq n/2$. This completes the proof. 
\end{proof}

\bibliographystyle{alpha}
\bibliography{main}

\end{document}